\documentclass[11pt]{article}
\usepackage[letterpaper, left=1in, right=1in, bottom=1in, top=1in]{geometry}
\usepackage{authblk}
\usepackage[T1]{fontenc}
\usepackage{graphicx}
\usepackage{amsmath,amssymb,amsthm}
\newcommand{\DC}{\mathcal{C}}
\newcommand{\lrs}{\mathit{lrs}}
\newcommand{\lrslrs}{\mathit{lrs}^{2}}
\newcommand{\bord}{\mathit{bord}}
\newcommand{\WAQ}{\mathsf{WAQ}}
\newcommand{\STree}{\mathsf{STree}}
\newcommand{\str}{\mathsf{str}}
\newcommand{\len}{\mathsf{strlen}}
\newcommand{\leaf}{\mathsf{leaf}}
\newcommand{\occ}{\mathit{occ}}
\newcommand{\OC}{\mathsf{OC}}
\newcommand{\out}{\mathsf{output}}

\newtheorem{theorem}{Theorem}
\newtheorem{corollary}{Corollary}
\newtheorem{lemma}{Lemma}
\newtheorem{observation}{Observation}
\newtheorem{proposition}{Proposition}

\begin{document}
\title{Online and Offline Algorithms for Counting Distinct Closed Factors via Sliding Suffix Trees}
\author[1]{Takuya Mieno}
\author[2]{Shun Takahashi}
\author[2]{Kazuhisa Seto}
\author[2]{Takashi Horiyama}
\affil[1]{The University of Electro-Communications, Japan}
\affil[ ]{\texttt{tmieno@uec.ac.jp}}
\affil[2]{Hokkaido University, Japan}
\affil[ ]{\texttt{\{seto,horiyama\}@ist.hokudai.ac.jp}}
\date{}
\maketitle              \begin{abstract}
  A string is said to be closed if its length is one, or
  if it has a non-empty factor that occurs both as a prefix and as a suffix of the string,
  but does not occur elsewhere.
  The notion of closed words was introduced by [Fici, WORDS 2011].
  Recently, the maximum number of distinct closed factors occurring in a string was investigated
  by [Parshina and Puzynina, Theor. Comput. Sci. 2024],
  and an asymptotic tight bound was proved.
  In this paper, we propose two algorithms
  to count the distinct closed factors
  in a string $T$ of length $n$ over an alphabet of size $\sigma$.
  The first algorithm runs in $O(n\log\sigma)$ time using $O(n)$ space
  for string $T$ given in an online manner.
  The second algorithm runs in $O(n)$ time using $O(n)$ space
  for string $T$ given in an offline manner.
  Both algorithms utilize suffix trees for sliding windows.
\end{abstract}
\section{Introduction}\label{sec:intro}
String processing is a fundamental area in computer science,
with significant importance ranging from theoretical foundations to practical applications.
One of the most active areas of this field is the study of repetitive structures within strings,
which has driven advances in areas such as 
pattern matching algorithms~\cite{KnuthMP77,GalilS83} and
compressed string indices~\cite{BelazzouguiCGGK21,KempaP18,KociumakaNO24}.
Understanding repetitive structures in strings is important
for the advancement of information processing technology.
For surveys on these topics, see~\cite{CrochemoreIR09,Smyth13} and~\cite{Navarro21,Navarro21a}.
The concept of \emph{closed words}~\cite{Fici2011}
is a sort of such repetitive structures of strings.
A string is said to be \emph{closed} if its length is one, or
if it has a non-empty factor that occurs both as a prefix and as a suffix of the string,
but does not occur elsewhere\footnote{The notion of closed words is equivalent to those of return words~\cite{Durand98,GlenJWZ09} and periodic-like words~\cite{CarpiL01}.}.
For example, string $\mathtt{abaab}$ is closed because $\mathtt{ab}$ occurs both as a prefix and as a suffix, but does not occur elsewhere in $\mathtt{abaab}$.
Closed words have been studied primarily in the field of combinatorics on finite and infinite words~\cite{Vuillon01,BucciLL09,Fici2011,LucaF13,BadkobehFL15,fici2017open,ParshinaZ19,Badkobeh0FP22,ParshinaP24}.
Regarding the number of closed \emph{factors} (i.e., substrings) appearing in a string,
an asymptotic tight bound $\Theta(n^2)$ on the maximum number of distinct closed factors of a string is known~\cite{BadkobehFL15}.
More recently, Parshina and Puzynina refined this bound to $\sim \frac{n^2}{6}$ in 2024~\cite{ParshinaP24}.
Despite these progresses on the number of closed factors,
there is no non-trivial algorithm for computing the exact number of distinct closed factors of a given string to our knowledge.

In this paper, we present both online and offline algorithms for counting the number of distinct closed factors of a given string $T$ of length $n$.
The first counting algorithm is an online approach
performed in $O(n\log\sigma)$ time and $O(n)$ space,
where $\sigma$ is the number of distinct characters in the string.
The second counting algorithm is an offline approach
performed in $O(n)$ time and $O(n)$ space,
assuming $T$ is drawn from an integer alphabet of size $n^{O(1)}$.
We begin by characterizing the number of distinct closed factors of $T$
through the \emph{repeating suffixes} of some prefixes and factors of $T$.
Based on this characterization, we design an online algorithm that
utilizes Ukkonen's online suffix tree construction~\cite{Ukkonen95}, as well as suffix trees for a sliding window~\cite{FialaG89,Larsson96,Senft2005,LeonardarXiv}.
We then design a linear-time offline algorithm by simulating
sliding-window suffix trees within the static suffix tree of the entire string $T$.
This simulation is of independent interest,
as it has the potential to speed up sliding-window algorithms for strings in an offline setting.
Furthermore, we explore the enumeration of (distinct) closed factors in a string
and propose an algorithm
that combines our counting method with a geometric data structure
for handling points in the two-dimensional plane~\cite{BelazzouguiP16}, resulting a somewhat faster solution.

\paragraph*{Related work.}
Recent work has highlighted algorithmic advances
in the study of closed factors and related problems
over the past decade~\cite{BannaiIKLRRSW15,BadkobehBGIIIPS16,OCarray,AlzamelISS19,AlamroAIPSW20,Badkobeh0FP22,abs-2209-00271,Sumiyoshi2024}.
The line of algorithmic research on closed factors is initialized by Badkobeh et al,~\cite{BadkobehBGIIIPS14,BadkobehBGIIIPS16},
who addressed various problems related to factorizing a string into a sequence of closed factors and proposed efficient algorithms for these tasks.
In the domain of online string algorithms,
Alzamel et al.~\cite{AlzamelISS19} proposed an algorithm that computes closed factorizations in an online settings,
and more recently, Sumiyoshi et al.~\cite{Sumiyoshi2024} achieved a speedup of $\log \log n$ times in total execution time.
Additionally, a new concept, \emph{maximal closed substrings}~(MCS), was introduced in~\cite{Badkobeh0FP22},
along with an $O(n \log n)$-time enumeration algorithm~\cite{abs-2209-00271}.

\section{Preliminaries}\label{sec:prelim}
\subsection{Basic notations}
Let $\Sigma$ be an ordered \emph{alphabet}.
An element of $\Sigma$ is called a \emph{character}.
An element of $\Sigma^\star$ is called a \emph{string}.
The \emph{empty string}, denoted by $\varepsilon$, is the string of length $0$.
For a string $T \in \Sigma^\star$, the length of $T$ is denoted by $|T|$.
If $T = xyz$ for some strings $x,y,z\in\Sigma^\star$, we call $x$, $y$, and $z$ a \emph{prefix}, a \emph{factor}, and a \emph{suffix} of $T$, respectively.
A string $b \ne T$ is said to be a \emph{border} of $T$ if $b$ is both a prefix of $T$ and a suffix of $T$.
We denote by $\bord(T)$ the longest border of $T$.
Note that the longest border always exists since $\varepsilon$ is a border of any string.
For each $i$ with $1 \le i \le |T|$, we denote by $T[i]$ the $i$th character of $T$.
For each $i, j$ with $1 \le i \le j \le |T|$, we denote by $T[i.. j]$ the factor of $T$ that starts at position $i$ and ends at position $j$.
For convenience, we define $T[i.. j] = \varepsilon$ for any $i > j$.
For a strings $T$ and $S$, the set $\occ_T(S) = \{i\mid T[i.. i+|S|-1] = S\}$ of integers is said to be the \emph{occurrences} of $S$ in $T$.
If $\occ_T(S) \ne \emptyset$, we say that $S$ \emph{occurs} in $T$ as a factor.
A string $T$ is said to be \emph{closed} if there is a border of $T$ that occurs exactly twice in $T$.
If $|\occ_T(S)| = 1$, we say that $S$ is a \emph{unique} factor of $T$.
Also, if $|\occ_T(S)| \ge 2$, we say that $S$ is a \emph{repeating} factor of $T$.
We denote by $\lrs(T)$ the longest repeating suffix of string $T$.
Further we denote $\lrslrs(T) = \lrs(\lrs(T))$.

In what follows, we fix a non-empty string $T$ of arbitrary length $n$ over an alphabet $\Sigma$ of size $\sigma$.
This paper assumes the standard word RAM model with word size $\Omega(\log n)$.

\subsection{Suffix trees}\label{subsec:suftree}
The most important tool of this paper is a \emph{suffix tree} of string $T$~\cite{Weiner73}.
A suffix tree of $T$, denoted by $\STree(T)$, is a \emph{compact trie} for the set of suffixes of $T$.
Below, we summarize some known properties of suffix trees and define related notations: 
\begin{itemize}
  \item Each edge of $\STree(T)$ is labeled with a factor of $T$ of length one or more.
  \item Each internal node $u$, including the root, has at least two children, and
    the first characters of the labels of outgoing edges from $u$ are mutually distinct (unless $T$ is a unary string).
  \item For each node $v$ of $\STree(T)$, we denote by $\str_T(v)$ the string spelled out from the root to $v$.
    We also define $\len_T(u) = |\str_T(u)|$ for a node $u$.
  \item There is a one-to-one correspondence between leaves of $\STree(T)$ and unique suffixes of $T$.
    More precisely, for each leaf $\ell$, the string $\str_T(\ell)$ equals some unique suffix of $T$, and vice versa.
    Note that repeating suffixes of $T$ are not always represented by a node in $\STree(T)$.
\end{itemize}
Throughout this paper, we assume for convenience that the last character of $T$ is unique.
For each  $1 \le i \le n$, we denote by $\leaf_T(i)$ the leaf $\ell$ of $\STree(T)$ such that  $\str_T(\ell) = T[i.. n]$.

It is known that $\STree(T)$ for given string $T$ can be constructed in $O(n)$ time~\cite{Farach97}
if $\Sigma$ is \emph{linearly sortable}\footnote{A typical example of linearly sortable alphabets is an integer alphabet $\{1, 2, \ldots, n^c\}$ for some constant $c$. Any $n$ characters from the alphabet can be sorted in $O(n)$ time using radix sort.}, i.e.,
any $n$ characters from $\Sigma$ can be sorted in $O(n)$ time.
Additionally, when the input string $T$ is given in an \emph{online} manner,
$\STree(T)$ can be constructed in $O(n\log\sigma)$ time using Ukkonen's algorithm~\cite{Ukkonen95}.
In both algorithms, the resulting suffix trees are edge-sorted,
meaning that the first characters of the labels of outgoing edges from each node are (lexicographically) sorted.
Thus we assume that all suffix trees in this paper are edge-sorted.
\begin{theorem}[\cite{Ukkonen95}]\label{thm:ukkonen}
  For incremental $i = 1, 2, \ldots, n$,
  we can maintain the suffix tree of $T[1..i]$ and the length of 
  $\lrs(T[1.. i])$ in a total of $O(n \log \sigma)$ time.
\end{theorem}
In other words, for every $j = 1, 2, \ldots, n-1$,
we can update $\STree(T[1..j])$ to $\STree(T[1.. j+1])$ in amortized $O(\log\sigma)$ time.
Ukkonen's algorithm maintains the \emph{active point} in the suffix tree,
which is the locus of the longest repeating suffix of the string.
To maintain the active point efficiently,
Ukkonen's algorithm uses auxiliary data structures called \emph{suffix links}:
Each internal node $u$ of the suffix tree has a suffix link
that points to the node $v$, such that $\str_T(v)$ is the suffix of $\str_T(u)$ of length $|\str_T(u)|-1$.
See Fig.~\ref{fig:stree} for examples of a suffix tree and related notations.
\begin{figure}[tb]
  \centerline{
    \includegraphics[width=0.5\linewidth]{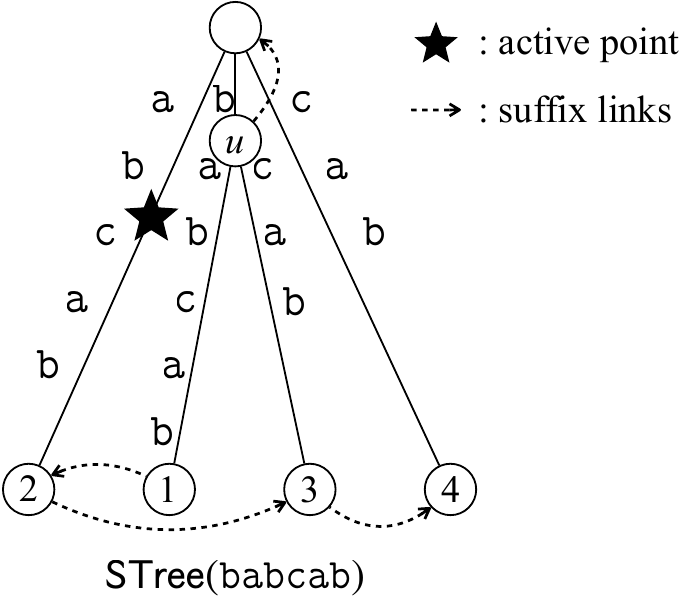}
  }
  \caption{The suffix tree of string $T = \mathtt{babcab}$.
    Each leaf of the tree represents a suffix of $T$, with the integer inside each leaf indicating the starting position of the suffix.
    Namely, the leaf labeled with number $i$ corresponds to $\leaf_T(i)$.
    In this suffix tree, $\str_T(u) = \mathtt{b}$, $\len_T(u) = 1$, $\str_T(\leaf_T(3)) = \mathtt{bcab}$, and $\len_T(\leaf_T(3)) = 4$.
    The star symbol indicates the locus of the active point, which represents the longest repeating suffix $\mathtt{ab}$ of $T$.
    The dotted arrows represent suffix links.
  }
  \label{fig:stree}
\end{figure}

Furthermore, based on Ukkonen's algorithm, we can maintain suffix trees for
a \emph{sliding window} over $T$ in a total of $O(n \log\sigma)$ time,
using space proportional to the window size:
\begin{theorem}[\cite{FialaG89,Larsson96,Senft2005,LeonardarXiv}]\label{thm:slidingST}
  Using $O(W)$ space, a suffix tree for a variable-width sliding window
  can be maintained in amortized $O(\log\sigma)$ time per one operation: either
  (1) appending a character to the right or
  (2) deleting the leftmost character from the window,
  where $W$ is the maximum width of the window.
  Additionally, the length of the longest repeating suffix of the window can be maintained with the same time complexity.
\end{theorem}
In other words, for every incremental pair of indices $i, j$ with $i \le j$,
we can update $\STree(T[i..j])$ to either $\STree(T[i.. j+1])$ or $\STree(T[i+1.. j])$
in amortized $O(\log\sigma)$ time.
Later, we use the above algorithmic results as black boxes.

\subsection{Weighted ancestor queries}
For a node-weighted tree $\mathcal{T}$, where the weight of each non-root node is greater than the weight of its parent,
a \emph{weighted ancestor query} (WAQ) is defined as follows:
given a node $u$ of the tree and an integer $t$,
the query returns the farthest (highest) ancestor of $u$ whose weight is at least $t$.
We denote by $\WAQ_{\mathcal{T}}(u, t)$ the returned node for weighted ancestor query $(u, t)$ on tree $\mathcal{T}$.
The subscript $\mathcal{T}$ will be omitted when clear from the context.
In general, a WAQ on a tree of size $N$ can be answered in $O(\log\log N)$ time,
which is worst-case optimal with linear space.
However, if the input is the suffix tree of some string of length $n$, and
the weight function is $\len_T$, as defined in the previous subsection,
any WAQ can be answered in $O(1)$ time after $O(n)$-time preprocessing~\cite{BelazzouguiKPR21}.

\section{Counting closed factors online}\label{sec:online}

In this section, we propose an algorithm to count the distinct closed factors
for a string given in an online manner.
Since there may be $\Omega(n^2)$ distinct closed factors in a string of length $n$~\cite{BadkobehFL15,ParshinaP24},
enumerating them requires quadratic time in the worst case.
However, for counting the closed factors, we can achieve subquadratic time
even when the string is given in an online manner, as described below.

\subsection{Changes in number of closed factors}
First we consider the changes in the number of distinct closed factors
when a character is appended.
Let $\DC(T)$ be the set of distinct closed factors occurring in $T$.
Let $d_j = |\DC(T[1.. j])| - |\DC(T[1.. j-1])|$ for $2\le j \le n$.
For convenience let $d_1 = 1$.

\begin{observation}\label{obs:bord}
  For any closed suffixes $u$ and $v$ of the same string,
  $|\bord(u)| \ne |\bord(v)|$ holds if $u \ne v$.
\end{observation}
Based on this observation,
we count the distinct borders of closed factors
instead of the closed factors themselves. 
We show the following:
\begin{lemma}\label{lem:num}
  For each $j$ with $1 \le j \le n$,
  $d_j = |\lrs(T[1.. j])| - |\lrslrs(T[1.. j]))|$ holds
  if $\lrs(T[1.. j]) \ne \varepsilon$.
  Also, $d_j = 1$ holds if $\lrs(T[1.. j]) = \varepsilon$.
\end{lemma}
\begin{proof}
  The latter case is obvious 
  since $\lrs(T[1.. j]) = \varepsilon$ implies that $T[i]$ is a unique character in $T[1.. j]$,
  making it the only new closed factor.
  In the following, we assume $\lrs(T[1.. j]) \ne \varepsilon$ and denote $t_j = \lrs(T[1.. j])$ and $z_j = \lrslrs(T[1.. j])$.
  It can be observed that for any new closed factor $u \in \DC(T[1..j])\setminus\DC(T[1.. j-1])$,
  $u$ is a unique suffix of $T[1.. j]$
  and the longest border of $u$ repeats in $T[1.. j]$.
  Equivalently, $|u| > |t_j|$ and $|\bord(u)| \le |t_j|$ hold.
  Additionally, $|u| > |t_j|$ implies $|\bord(u)| > |z_j|$.
  Thus $|z_j| < |\bord(u)| \le |t_j|$, and hence $d_j \le |t_j| - |z_j|$ by Observation~\ref{obs:bord}.
  Conversely, for any suffix $b$ of $T[1..j]$ satisfying $|z_j| < |b| \le |t_j|$,
  there exists exactly one closed suffix whose longest border is $b$
  since $b$ is repeating in $T[1..j]$.
  Further, since $|b| > |z_j|$, the closed suffix is longer than $t_j$,
  and thus the closed suffix is unique in $T[1.. j]$.
  Therefore, $d_j \ge |t_j| - |z_j|$ holds.
\end{proof}

We note that the upper bound
$d_j \le |\lrs(T[1.. j])| - |\lrslrs(T[1.. j]))|$ of $d_j$
is already claimed in~\cite{ParshinaP24}, however, the equality was not shown.
By definition of $d_j$ and Lemma~\ref{lem:num},
the next corollary immediately follows:
\begin{corollary}\label{cor:sum}
  The total number of distinct closed factors of $T$ is
  \[
    \sum_{j\in J}|\lrs(T[1.. j])| - \sum_{j\in J}|\lrslrs(T[1.. j]))| + n - |J|
  \]
  where $J = \{j \in [1, n] \mid \lrs(T[1.. j]) \ne \varepsilon\}$.
\end{corollary}

\subsection{Online algorithm}
By Theorem~\ref{thm:ukkonen},
we can compute $\sum_{j\in J}|\lrs(T[1.. j])|$ in $O(n\log\sigma)$ time online.
Thus, by Corollary~\ref{cor:sum}, our remaining task is to compute $\sum_{j\in J}|\lrslrs(T[1.. j])|$ online,
equivalently, to compute 
the sequence $\mathcal{Z} = (|z_1|, |z_2|$ $\ldots, |z_n|)$
where $z_j = \lrslrs(T[1.. j])$ for each $j$.
To compute the sequence,
we use sliding suffix trees of Theorem~\ref{thm:slidingST}.
It is known that the starting position of $\lrs(T[1.. j])$ is not smaller than
that of $\lrs(T[1.. j-1])$.
Namely, starting from $t_1 = \lrs(T[1.. 1]) = \varepsilon$,
the sequence $\langle t_1, t_2, \ldots, t_n\rangle$ of factors of $T$
can be obtained by $O(n)$ \emph{sliding-operations} on $T$ that consists of 
(1) appending a character to the right or
(2) deleting the leftmost character from the factor,
where $t_j = \lrs(T[1.. j])$ for each $j$.
Thus, by Theorem~\ref{thm:slidingST}, we can obtain every $|z_i| \in \mathcal{Z}$ in amortized $O(\log\sigma)$ time
by appropriately applying sliding-operations
to reach the windows $t_j$ for all $j$ with $j \in J$ in order.
Finally, we obtain the following:
\begin{theorem}\label{thm:online}
  For a string $T$ of length $n$ given in an online manner,
  we can count the distinct closed factors of $T$
  in $O(n\log\sigma)$ time using $O(n)$ space
  where $\sigma$ is the alphabet size.
\end{theorem}

\section{Counting closed factors offline: shaving log factor}\label{sec:offline}

In this section, we assume that the alphabet is linearly sortable,
allowing us to construct $\STree(T)$ in $O(n)$ time offline~\cite{Farach97}.
Our goal is to remove the $\log \sigma$ factor from the complexity in Theorem~\ref{thm:online},
which arises from using dynamic predecessor dictionaries (e.g., AVL trees) at non-leaf nodes.
To shave this logarithmic factor, we do not rely on such dynamic dictionaries.
Instead, use a WAQ data structure constructed over $\STree(T)$.

\paragraph*{\bf Simulating sliding suffix tree in linear time.}
The idea is to simulate the online algorithm from Section~\ref{sec:online} using the suffix tree of the entire string $T$.
In algorithms underlying Theorem~\ref{thm:ukkonen} and~\ref{thm:slidingST},
the topology of suffix trees and the loci of the active points change step-by-step.
The changes include
(i) moving the active point, 
(ii) adding a node, an edge, or a suffix link, and
(iii) removing a node, an edge, or a suffix link.

Our data structure consists of two suffix trees:
\begin{itemize}
  \item[(1)] The suffix tree $\STree(T)$ of $T$ enhanced with a WAQ data structure.
  \item[(2)] The suffix tree $\STree(w)$ of the factor $w = T[i.. j]$, which represents a sliding window.
\end{itemize}
Note that $\STree(w)$ includes the active point.
All nodes and edges of $\STree(w)$ are connected to their corresponding nodes and edges in $\STree(T)$.
Specifically, an internal node $u$ of $\STree(w)$ is connected to the node $\tilde{u}$ of $\STree(T)$ where $\str_w(u) = \str_T(\tilde{u})$.
Similarly, an edge $(u, v)$ of $\STree(w)$ is connected to the edge $(\tilde{u}, v')$ of $\STree(T)$
where the first characters of the edge-labels are the same.
Such node $\tilde{u}$ and edge $(\tilde{u}, v')$ in $\STree(T)$ always exist:
An internal node $u$ implies that, for some distinct characters $c_1$ and $c_2$,
$\str_w(u)c_1$ and $\str_w(u)c_2$ occur in $w$, and thus also in $T$.
Furthermore, a leaf $\ell$ of $\STree(w)$, representing suffix $T[k.. j]$ of $w$,
is connected to the leaf of $\STree(T)$ representing the suffix $T[k.. n]$ of $T$.
Note that suffix links of $\STree(w)$ do not connect to $\STree(T)$ and are maintained within $\STree(w)$.
See Fig.~\ref{fig:offline} for illustration.
\begin{figure}[tb]
  \centerline{
    \includegraphics[width=\linewidth]{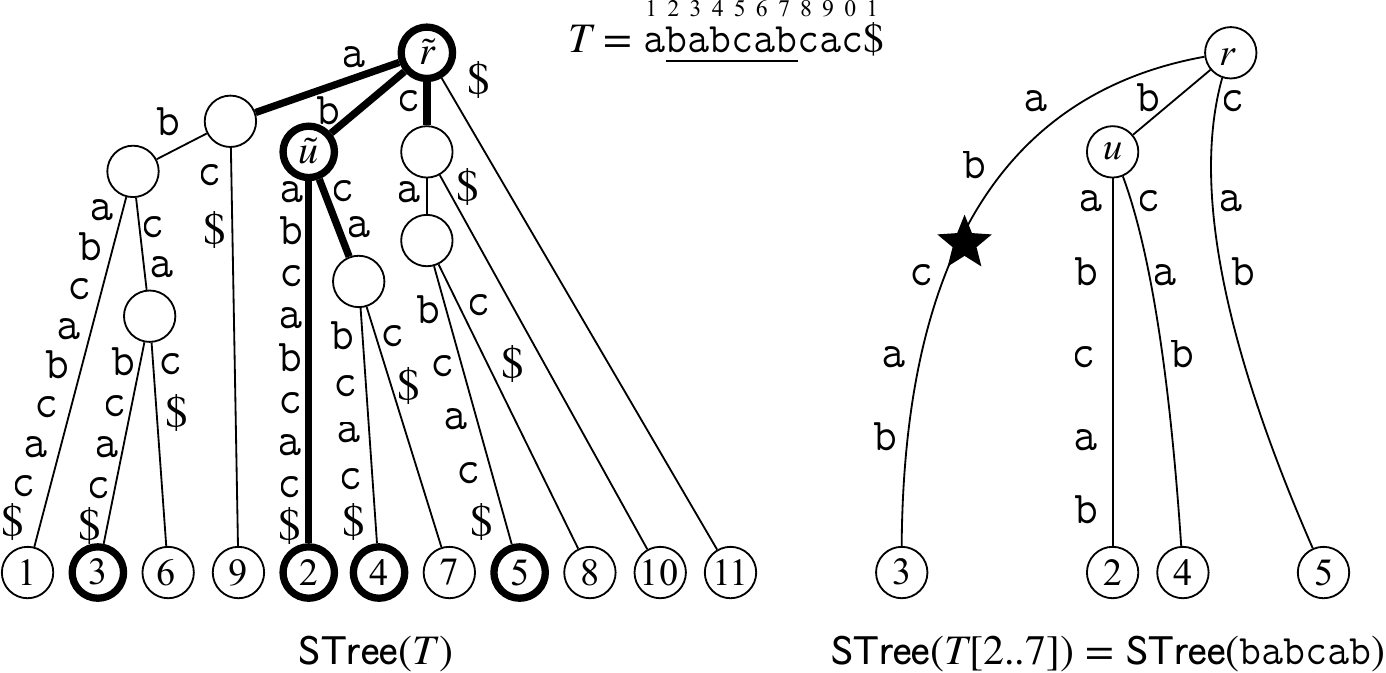}
  }
  \caption{
    The suffix tree of string $T = \mathtt{ababcabcac\$}$ is shown on the left.
    The suffix tree of $T[2.. 7] = \mathtt{babcab}$, which is the same as the one in Fig.~\ref{fig:stree}, is shown on the right.
    Suffix links are omitted in this figure.
    In the tree on the left, each node and edge enclosed in a bold line is connected to a corresponding one in the tree on the right.
    For clarity, those connections are not drawn. 
  }
  \label{fig:offline}
\end{figure}

We maintain the active point in $\STree(w)$ as follows:
\begin{itemize}
  \item When the active point is on an edge and moves down,
    if the character on the edge following the active point matches the next character $T[j+1]$,
    the active point simply moves down.
    Otherwise, if the active point cannot move down, we create a new branching node $u$ at the locus and add a new leaf $\ell$ and edge $e = (u, \ell)$.
    The new leaf and edge are connected to the corresponding edge $\tilde{e}$ and leaf of $\STree(T)$ as the above discussion.
    The edge $\tilde{e}$ of $\STree(T)$ can be found in constant time using a weighted ancestor query on $\STree(T)$.
  \item Similarly, when the active point is a node $u$ and moves down,
    if $u$ represents $T[i'.. j]$, we
    query $\WAQ_{\STree(T)}(\leaf_T(i'), j-i'+2)$ and check the resulting edge.
    The resulting edge of $\STree(T)$ is already connected to an edge of the smaller suffix tree
    if and only if $u$ has an outgoing edge whose edge label starts with $T[j+1]$.
    If the active point can move down to the existing edge, we simply move it to the edge (via its corresponding edge in $\STree(T)$).
    Otherwise, we create a new leaf $\ell$ and edge $e = (u, \ell)$, connecting them to their corresponding leaf and edge of $\STree(T)$ as described above.
\end{itemize}

In Ukkonen's algorithm, a new node/edge is created only if the active point reaches it,
so every node/edge creation can be simulated as described.
Node/edge deletions are straightforward: 
we simply disconnect the node/edge from $\STree(T)$ and remove them.
Therefore, we can simulate the sliding suffix tree over $T$ in amortized $O(1)$ time per sliding-operation.
We have shown the next theorem:

\begin{theorem}\label{thm:simulatesliding}
  Given an offline string $T$ over a linearly sortable alphabet,
  we can simulate a suffix tree for a sliding window in a total of $O(n)$ time.
\end{theorem}

By Theorem~\ref{thm:simulatesliding} and the algorithm described in Section~\ref{sec:online}, we obtain the following:
\begin{corollary}\label{cor:offline}
  We can count the distinct closed factors in $T$ over a linearly sortable alphabet in $O(n)$ time using $O(n)$ space.
\end{corollary}

\section{Enumerating closed factors offline}\label{sec:enum}
In this section, we discuss the enumeration of (distinct) closed factors in an offline string $T$.
We first consider The \emph{open-close-array} $\OC_T$ of $T$, which is the binary sequence of length $n$ where $\OC_T[i] = 1$ iff $T[i.. n]$ is closed~\cite{OCarray}.
Since an open-close-array can be computed in $O(n)$ time, we can enumerate all the occurrences of closed factors of $T$ in $\Theta(n^2)$ time
by constructing $\OC_{T[i..n]}$ for all $1 \le i \le n$.
We then map the occurrences to their corresponding loci on $\STree(T)$ using constant-time WAQs $O(n^2)$ times.
Thus we have the following:
\begin{proposition}
  We can enumerate all the occurrences of closed factors of $T$ and the distinct closed factors of $T$ in $\Theta(n^2)$ time.
\end{proposition}

Next, we propose an alternative approach that can achieve subquadratic time when there are few closed factors to output.
Using the algorithms described in Sections~\ref{sec:online} and~\ref{sec:offline},
we can enumerate the ending positions and the \emph{longest borders} of distinct closed factors in $O(n + \out)$ time.
For each such a border $b$ of a closed factor,
we need to determine the starting position of the closed factor
by finding the nearest occurrence of $b$ to the left.
Such occurrences can be computed efficiently by utilizing
\emph{range predecessor queries} over the list of (the integer-labels of) the leaves of $\STree(T)$ as follows:
Let $b = T[s.. t]$ be a border of some closed factor $w$ where $b$ occurs exactly twice in $w$.
We first find the locus of $b$ in $\STree(T)$ by a WAQ.
Let $v$ be the nearest descendant of the locus, inclusive.
Then, the leaves in the subtree rooted at $v$ represent the occurrences of $b$ in $T$.
Thus the nearest occurrence of $b$ to the left equals
the predecessor value of $s$ within the leaves under $v$.
By applying Belazzougui and Puglisi's range predecessor data structure~\cite{BelazzouguiP16}, we obtain the following:
\begin{proposition}
  We can enumerate the distinct closed factors of $T$ in $O(n\sqrt{\log n} + \out\cdot \log^\epsilon n)$ time
  where $\out \in O(n^2)$ is the number of distinct closed factors in $T$ and
  $\epsilon$ is a fixed constant with $0 < \epsilon < 1$.
\end{proposition}

\section{Conclusions}\label{sec:conclusion}
In this paper, we proposed two algorithms for counting
distinct closed factors of a string.
The first algorithm runs in $O(n \log \sigma)$ time using $O(n)$ space
for a string given in an online manner.
The second algorithm runs in $O(n)$ time and space
for a static string over a linearly sortable alphabet.
Additionally, we discussed how to enumerate the distinct closed factors,
and showed a $\Theta(n^2)$-time algorithm using open-close-arrays, as well as
an $O(n\sqrt{\log n} + \out\cdot\log^\epsilon n)$-time algorithm that combines our counting algorithm with a range predecessor data structure.
Since there can be $\Omega(n^2)$ distinct closed factors in a string~\cite{BadkobehFL15,ParshinaP24},
the $\out\cdot\log^\epsilon n$ term can be superquadratic in the worst case.
This leads to the following open question:
For the enumerating problem, can we achieve $O(n\mathsf{polylog}(n) + \out)$ time, which is linear in the output size?

\bibliographystyle{plain}

\end{document}